\documentclass[12pt]{amsart}
\usepackage{tikz}
\usetikzlibrary{positioning}
\usepackage{hhline}
\usepackage{xcolor}

\usepackage{braket}
\usepackage{graphicx} 
\usepackage[margin=1in]{geometry}
\usepackage{mathrsfs}

\usepackage[
    backend=biber,
    style=numeric,
    sorting=none,
    url=false,
    doi=false,
    isbn=false,
    eprint=false
]{biblatex}
\usepackage{latexsym}
\usepackage{amsmath,amssymb,amsfonts}

\usepackage{graphicx}
\usepackage{wrapfig}
\usepackage{fancybox}
\usepackage{bm}

\usepackage{amsmath}
\usepackage{mathdots}
\usepackage{amsthm}
\usepackage{tikz-cd}
\usepackage{thmtools}
\declaretheorem[style=plain, numberwithin=section]{theorem}
\declaretheorem[style=definition,name=Definition,qed=$\blacksquare$, numberwithin=section, sibling=theorem]{definition}

\newtheorem{lemma}[theorem]{Lemma}

\newtheorem{corollary}[theorem]{Corollary}

\usepackage{url}

\renewcommand{\epsilon}{\varepsilon}%

\newcommand{\SL}{\operatorname{SL}}
\newcommand{\GL}{\operatorname{GL}}
\newcommand{\SU}{\operatorname{SU}}

\newcommand{\U}{\operatorname{U}}

\newcommand{\CC}{\mathbb C}

\newcommand{\ZZ}{\mathbb Z}

\usepackage{color}

\newcommand{\im}{\mathrm i}

\title[Transversal gates and local symmetries]{Transversal gates of the ((3,3,2)) qutrit code and local symmetries of the absolutely maximally entangled state of four qutrits}
\author{Ian Tan$^*$}
\thanks{$^*$ Corresponding author. Affiliation: Department of Mathematics and Statistics, Auburn University, Auburn, AL, USA. E-mail: \href{mailto:yzt0060@auburn.edu}{yzt0060@auburn.edu}}

\usepackage{xcolor}
\definecolor{darkgreen}{rgb}{0,0.65,0}

\usepackage[
    colorlinks=true,
    linkcolor=blue,
    citecolor=darkgreen,
    urlcolor=magenta
]{hyperref}
\usepackage{cleveref}

\begin{document}

\begin{abstract}
    The group of transversal gates and the group of local symmetries are important features of quantum error correcting codes and pure quantum states, respectively; the former provides fault-tolerant operations on a code while the latter tells us about a state's reachability via stochastic local operations with classical communication. We prove that there exists a bijection between local unitary (LU) orbits of absolutely maximally entangled (AME) states in $(\mathbb{C}^D)^{\otimes n}$ where $n$ is even, also known as perfect tensors, and LU orbits of $((n-1,D,n/2))_D$ quantum error correcting codes. Furthermore, there is a close connection between the local symmetries of an AME state and the transversal gates of its corresponding quantum error correcting code. We explore in detail the 4-qutrit AME state $\ket{\Phi}$ and its corresponding $((3,3,2))_3$ qutrit code $\mathcal{C}$. We show that $\ket{\Phi}$ and $\mathcal{C}$ are both unique up to the action of the LU group. We find generators of the local symmetry group of $\ket{\Phi}$ and the group of transversal gates on $\mathcal{C}$. Our proofs rely on prior results by Huber and Grassl (2020), Hebenstreit et al. (2016), and Rather et al. (2023). We use Vinberg's theory of Lie algebras to study the special case of $\ket{\Phi}$ and $\mathcal{C}$.
\end{abstract}
\maketitle

\section{Introduction}
Quantum error correction is a critical safeguard for any robust quantum protocol, as it preserves the integrity of fragile qudits (a \textit{qudit} is a $D$-level quantum system; if $D=2$ or $D=3$, we call them \textit{qubits} or \textit{qutrits} respectively) against environmental noise to ensure that computations and communications remain reliable. A quantum error correcting code (QECC) is a carefully selected subspace of a larger Hilbert space, where a logical qudit is embedded across multiple physical qudits to shield it from local errors. The transversal gates of a QECC are especially useful, as they are implemented by local unitaries and thus enable fault-tolerant operations on encoded qudits that do not propagate local errors.

Entanglement is a fundamental resource that enables many quantum information protocols to outperform their classical counterparts. An important class of highly entangled quantum states is the absolutely maximally entangled (AME) states. These AME states have found applications in quantum secret sharing \cite{PhysRevA.86.052335}, open-destination teleportation \cite{Helwig:2013qoq}, quantum information masking \cite{PhysRevA.104.032601}.

This paper explores a connection between AME states of an even number of qudits, also known as perfect tensors, and particularly nice QECCs known as maximum distance separable. It begins with a construction of Rains which can be used to produce a quantum error correcting code from an $r$-uniform state. In subsequent work, Huber and Grassl show that, for perfect tensors, this construction can always be reversed. We expand on this idea, first showing that the construction gives a one-to-one correspondence between local unitary (LU) orbits of perfect tensors and LU orbits of quantum error correcting codes. Second, we show that there exists a close relationship between the invertible local operations that fix an $r$-uniform state, known as local symmetries, and the transversal gates of its corresponding code. This adds to the practical consequences of a state's local symmetries, which is known to determine the state's reachability: a state $\ket{\psi}$ is \textit{reachable} if there is a state not equivalent by LU operations that can be deterministically converted into $\ket{\psi}$ by local operations with classical communication \cite{PhysRevLett.118.040503}. After establishing these general results, we turn to a particular example: the 4-qutrit AME state $\ket{\Phi}$ and its corresponding $((3,3,2))_3$ code $\mathcal{C}$. This case is especially interesting because of a connection to Vinberg's theory of graded Lie algebras. We show that the state $\ket{\Phi}$ and the code $\mathcal{C}$ are both unique up to the action of the LU group. Then we calculate the local symmetries of the 4-qutrit AME state and the transversal gates of the $((3,3,2))_3$ code.

We wish to say more about the relevance of Vinberg's theory of graded Lie algebras. The crucial fact is that the 4-qubit state space $(\CC^2)^{\otimes 4}$ and the 3-qutrit state space $(\CC^3)^{\otimes 3}$ can both be embedded into graded Lie algebras so that results from Vinberg's paper \cite{Vinberg_1976} may be applied. For example, Vinberg's theory has been used to extend to other cases the observation of D\"ur et al. that three qubits can be entangled in different ways \cite{PhysRevA.62.062314}. More specifically, Vinberg's theory allows us to classify the families of pure states that are equivalent by stochastic local operations with classical communication (SLOCC) for four qubits \cites{ChtDjo:NormalFormsTensRanksPureStatesPureQubits,dietrich2022classification} and three qutrits \cites{Nurmiev_2000,di2023classification}. More recent work uses Vinberg's theory to find 4-qubit and 3-qutrit states that are highly entangled with respect to LU invariant polynomials \cite{Jaffali_2024,4qubit}. However, connections between Vinberg's theory and QECCs seem to have been overlooked. To our knowledge, it has not been previously mentioned in the literature that the unique $((4,4,2))_2$ and $((3,3,2))_3$ codes are Cartan subspaces or that their associated Weyl groups are their transversal gate groups. We attempt to bring more exposure facts like these and more broadly to consolidate various results at the intersection of geometric invariant theory (GIT) and quantum information theory (QIT).

For now, we assume that the reader is familiar with certain concepts and definitions used in the study of QECCs and Vinberg's theory of graded Lie algebras. We point out the following facts, which we prove later in the text.
\begin{enumerate}
    \item[1.] $\mathcal{C}$ is simultaneously a $((3,3,2))_3$ code and a Cartan subspace. \label{item:1}
    \item[2.] The stabilizer of $\mathcal{C}$ (in the QECC sense) is the centralizer of $\mathcal{C}$ (in Vinberg's sense). \label{item:2}
    \item[3.] The group of transversal gates on $\mathcal{C}$ is the Weyl group of $\mathcal{C}$. \label{item:3}
    \item[4.] Every $((3,3,2))_3$ code is $\SU_3^{\otimes 3}$-equivalent to $\mathcal{C}$ while every Cartan subspace of $(\CC^3)^{\otimes 3}$ is $\SL_3^{\otimes 3}$-equivalent to $\mathcal{C}$. \label{item:4}
    \item[5.] Every 3-qutrit AME state is $\SU_3^{\otimes 3}$-equivalent to a point in $\mathcal{C}$ while every semisimple element is $\SL_3^{\otimes 3}$-equivalent to a point in $\mathcal{C}$. \label{item:5}
\end{enumerate}

The first item is the key observation that brings together Vinberg's theory and the theory of QECCs in our setting. Note that $\mathcal{H}_{333}=(\CC^3)^{\otimes 3}$ is the grade-1 subspace of the $\ZZ_3$-graded Lie algebra
\begin{equation}\label{eq:z3grading}
\mathfrak{e}_6\cong \mathfrak{sl}_3^{\times 3}\oplus \mathcal{H}_{333}\oplus \mathcal{H}_{333}^*
\end{equation}
This makes sense of calling $\mathcal{C}\subset\mathcal{H}_{333}$ a Cartan subspace.

The second item is equivalent to
\begin{equation}\label{eq:first}
    \{g\in \mathcal{P}_3:g\ket{\varphi}=\ket{\varphi}\text{ for all }\ket{\varphi}\in\mathcal{C}\}=\{g\in\SL_3^{\otimes 3}:g\ket{\varphi}=\ket{\varphi}\text{ for all }\ket{\varphi}\in\mathcal{C}\},
\end{equation}
where $\mathcal{P}_3$ is the Pauli group on 3 qutrits. This fact follows from a lengthy algebraic computation by Hebenstreit et al. \cite{Hebenstreit}, although they do not state their result in the language of Vinberg or QECCs.

The third item is closely related to \eqref{eq:first} and is equivalent to
\begin{equation*}
    \{g\in \SU_3^{\otimes 3}:g\mathcal{C}=\mathcal{C}\}=\{g\in\SL_3^{\otimes 3}:g\mathcal{C}=\mathcal{C}\}.
\end{equation*}
Note that $\mathcal{P}_3\leq \SU_3^{\otimes 3}$. The proof relies on  \eqref{eq:first} and calculating normalizer elements that correspond to generators of the Weyl group. This simultaneously gives generators of the group of transversal gates, as it is the same as the Weyl group. Having done this, without much additional effort we also find generators of the group of local symmetries of $\ket{\Phi}$.

 The fourth item makes two analogous statements; the first one about $((3,3,2))_3$ codes follows from Rather et al.'s result \cite{PhysRevA.108.032412} on the uniqueness of $\ket{\Phi}$ while the second one about Cartan subalgebras is a basic result from Vinberg's theory. The fifth item also makes two analogous statements; the first is essentially a fact from GIT stated in the language of QIT while the second is another basic result from Vinberg's theory.

In \Cref{sec:QECC_uniform} we go over the necessary preliminaries on QECCs and uniform states. We define QECCs, stabilizer codes, the Pauli group $\mathcal{P}_n$, transversal gates, $r$-uniform states, AME states, and perfect tensors. In \Cref{sec:correspondence} we show that there exists a bijection between LU orbits of normalized perfect tensors and LU orbits of QECCs with certain parameters. We also state the connection between the local symmetries of an $r$-uniform state and the transversal gates of its corresponding code. Then we discuss connections to prior work. We are particularly interested in the perfect tensor $\ket{\Phi}$ and the code $\mathcal{C}$, which are written explicitly in terms of a fixed basis in \Cref{sec:332code}. In \Cref{sec:GIT} we go over the necessary preliminaries from GIT, including the Kempf-Ness theorem and Vinberg's theory of graded Lie algebras. We discuss many, but not all, aspects of Vinberg's theory that have applications in QIT. In particular, we do not discuss the decomposition of states into semisimple and nilpotent components, which is essential in the classification of SLOCC classes \cites{ChtDjo:NormalFormsTensRanksPureStatesPureQubits,dietrich2022classification,Nurmiev_2000,di2023classification}. We describe in detail how $\mathfrak{e}_6$ is a $\ZZ_3$-graded Lie algebra and how the $\SL_3^{\otimes 3}$-module $(\CC^3)^{\otimes 3}$ is the grade-1 subspace of $\mathfrak{e}_6$. \Cref{sec:mainresults} contains our results on $\ket{\Phi}$ and $\mathcal{C}$. We prove \eqref{eq:first}, give generators of the group of transversal gates on $\mathcal{C}$, and give generators of the group of local symmetries of $\ket{\Phi}$.

\section{Quantum error correcting codes and uniform states}\label{sec:QECC_uniform}
\subsection{Some important groups}\label{sec:somegroups} Let $\mathcal{H}=\CC^{d_1}\otimes\dots\otimes \CC^{d_n}$ be the state space of an $n$-partite system with local dimensions $d_i$. In this paper, a \textit{state} is simply a vector $\ket{\varphi}\in\mathcal{H}$. We say that a state is \textit{normalized} if we wish to impose the condition $\braket{\varphi|\varphi}=1$.

We are interested in various subgroups of $\GL_{d_1}\otimes\dots\otimes\GL_{d_n}$ that act on $\mathcal{H}$. For example, the subgroups \[G=\SL_{d_1}\otimes\dots\otimes\SL_{d_n}\quad \text{and}\quad K=\SU_{d_1}\otimes\dots\otimes\SU_{d_n}\] are featured in the Kempf-Ness theorem. It is known that two points of the projectivization $\mathbb{P}\mathcal{H}$ are interchangeable by SLOCC if and only if one can be transformed to the other by an element of $G$ \cite{PhysRevA.62.062314}.
Another important subgroup is the LU group $\U_{d_1}\otimes\dots \otimes\U_{d_n}$. We say that two states or subspaces of $\mathcal{H}$ are \textit{LU equivalent} if one can be transformed to the other by an element of the LU group. Note that two subspaces of $\mathcal{H}$ are LU equivalent if and only if they are equivalent by the action of $K$. This is because subspaces are fixed by a global phase change.

\subsection{Quantum error correcting codes}\label{sec:qecc} In \Cref{sec:qecc,sec:stabilizercodes,sec:transversal} we briefly introduce some concepts in the study of QECCs. More details can be found in references such as \cites{Gottesmanthesis,Nielsen_Chuang,Bacon_2013}. 

Let $\mathcal{B}$ be an orthogonal basis of $\text{End}(\CC^D)$ such that $\mathbb{I}\in\mathcal{B}$. Then we define a basis of $\text{End}((\CC^D)^{\otimes n})$ by
$
\mathcal{E}=\{E_1\otimes\dots\otimes E_n:E_i\in\mathcal{B},\: 1\leq i\leq n\}.
$
The elements $E\in\mathcal{E}$ represent errors that a code tries to correct. The \textit{weight} $\text{wt}(E)$ of an error $E=E_1\otimes\dots\otimes E_n$ is equal to the number $|\{i:E_i\neq \mathbb{I}\}|$ of single-qudit subsystems on which $E$ acts nontrivially.

\begin{definition}
    Let $\mathcal{Q}$ be a $K$-dimensional subspace of $(\CC^D)^{\otimes n}$ with an orthogonal basis $\{\ket{u_i}\}$. We say that $\mathcal{Q}$ is a (quantum error correcting) \textit{code} with parameters $((n,K,d))_D$ if $\bra{u_i}E\ket{u_j}=c(E)\delta_{ij}$ for all $E\in\mathcal{E}$ with $\text{wt}(E)<d$, where $c(E)$ is a constant depending on $E$. We say that $\mathcal{Q}$ is \textit{pure} if $c(E)=0$ whenever $0<\text{wt}(E)<d.$
\end{definition}
The parameter $d$ is the \textit{distance} of the code and represents the amount of error the code can correct: a code with distance $d\geq 2t+1$ can correct errors that affect up to $t$ subsystems. The parameters of a code satisfy the following inequality, known as the \textit{quantum Singleton bound}.

\begin{theorem}[Quantum Singleton bound]\label{thm:singleton}
    If $\mathcal{Q}$ is a code with parameters $((n,K,d))_D$, then $\log_D K\leq n-2(d-1)$.
\end{theorem}
\begin{proof}
    See \cite[{Theorem 2}]{782103} and \cite[{Theorem 13}]{Huber2020}.
\end{proof}
\begin{theorem}\label{thm:QMDS-pure}
    If $\mathcal{Q}$ is a code with parameters $((n,K,d))_D$ such that $\log_D K= n-2(d-1)$, then $\mathcal{Q}$ is pure.
\end{theorem}
\begin{proof}
    See \cite[{Theorem 2}]{782103} and \cite[{Theorem 5}]{Huber2020}.
\end{proof}
Codes that saturate the quantum Singleton bound are called \textit{maximum distance separable} (MDS). Thus, \Cref{thm:QMDS-pure} states that MDS codes are pure.

\subsection{Stabilizer codes}\label{sec:stabilizercodes} We now discuss a well-studied class of quantum error correcting codes known as stabilizer codes or linear codes. We introduce only the basic concepts needed to define a stabilizer code $\mathcal{Q}\subset(\CC^D)^{\otimes n}$ in the case where $D$ is prime. For a thorough treatment, see \cite{Nielsen_Chuang} for qubit codes and \cite{1715533} for the generalization to qudits.

For a prime $D$, let $\xi=e^{2\pi\im/D}$ be the $D$th root of unity and let $X,Z\in\text{End}(\CC^D)$ denote the \textit{generalized Pauli matrices} defined by
\begin{equation*}
X\ket{i}=\ket{(i+1) \text{ mod } D}\quad\text{and}\quad Z\ket{i}=\xi^i\ket{i}.
\end{equation*}
Note that $X$ and $Z$ depend on a chosen prime $D$ which will be clear from context. To give an example, when $D=3$, we have
\begin{equation}\label{eq:Pauli}
X = \begin{pmatrix} 0 & 0 & 1\\
1 & 0 & 0\\
0 & 1 & 0
\end{pmatrix}
\quad \text{and}\quad Z =\begin{pmatrix}
    1 & 0 & 0 \\
    0 & e^{2\pi \im/3} & 0\\
    0 & 0 & e^{4\pi \im/3}
\end{pmatrix}.
\end{equation}
The generalized Pauli matrices allow us to construct the \textit{Pauli group on $n$ qudits} $\mathcal{P}_n$. For $n=1$ qudit, the Pauli group $\mathcal{P}_1$ consists of all operators of the form $\xi^k U$, where $U\in\langle X,Z\rangle$ and $k\in\ZZ$. For general $n$, the Pauli group $\mathcal{P}_n$ consists of all $n$-fold tensor products of elements in $\mathcal{P}_1$.

If $\mathcal{S}$ is a subgroup of $\mathcal{P}_n$, define $V_\mathcal{S}$ to be the vector space consisting of all $\ket{\varphi}\in (\CC^D)^{\otimes n}$ such that $g\ket{\varphi}=\ket{\varphi}$ for all $g\in\mathcal{S}$. We say that $\mathcal{Q}$ is a \textit{stabilizer code} if it is equal to $V_\mathcal{S}$ for some $\mathcal{S}\leq \mathcal{P}_n$. In this case, the group $\mathcal{S}$ is the \textit{stabilizer} of $\mathcal{Q}$.

\subsection{Transversal gates}\label{sec:transversal} Let $\mathcal{Q}\subset (\CC^D)^{\otimes n}$ be a code. It is sometimes natural to consider not only operations from the Pauli group or LU group, but all SLOCC on the space $(\CC^D)^{\otimes n}$. This leads to the definition
\[
N(\mathcal{Q})=\{g\in\SL_D^{\otimes n}:g\mathcal{Q}=\mathcal{Q}\}.
\]
In some special cases, the group above is the normalizer of a Cartan subspace, which explains our choice of notation (more on this in \Cref{sec:GIT}). There is a natural representation 
\begin{equation}\label{eq:mu}
    \mu:N(\mathcal{Q})\to \GL(\mathcal{Q}),\quad\mu(g)\ket{\varphi}=g\ket{\varphi}
\end{equation}
for $g\in N(\mathcal{Q})$ and $\ket{\varphi}\in\mathcal{Q}$. If we fix an orthogonal basis $\{\ket{u_i}:1\leq i\leq K\}$ of $\mathcal{Q}$, then $\mu(g)$ is a $K\times K$ matrix with entries
$\mu(g)_{ij}=\bra{u_i}g\ket{u_j}$.

If $g\in N(\mathcal{Q})\cap \SU_D^{\otimes n}$, then we say that $\mu(g)$ is a \textit{transversal gate} of the code $\mathcal{Q}$. The subgroup $N(\mathcal{Q})\cap \SU_D^{\otimes n}$ gives us all LU operations that physically implement these gates. Transversal gates are of interest because they do not propagate local errors, enabling fault-tolerant operations on encoded qudits. Note that we only consider special unitary local operations $g$ in our definition of a transversal gate, ignoring the action of global phase changes.

\subsection{Uniform states} We now define and discuss some properties of a class of highly entangled states which generalize maximal entanglement in bipartite systems.

\begin{definition}
    A state $\ket{\varphi}\in (\CC^D)^{\otimes n}$ is \textit{$r$-uniform} if the reduction of $\ket{\varphi}\bra{\varphi}$ to any $r$-qudit subsystem is maximally mixed, that is, $\text{Tr}_S(\ket{\varphi}\bra{\varphi})$ is proportional to the identity for every subset $S\subset\{1,\dots,n\}$ of size $|S|= n-r$.
\end{definition}

Note the following facts about $r$-uniform states:
\begin{itemize}
    \item[1.] If $\ket{\varphi}$ is $r$-uniform, then $r\leq \lfloor \frac{n}{2}\rfloor$.
    \item[2.] If $\ket{\varphi}$ is $r$-uniform then $\ket{\varphi}$ is $r'$-uniform for every $1\leq r'\leq r$.
\end{itemize}
The first item follows from Schmidt decomposition with respect to a subsystem of size $r$ and its complement. The second item holds because a reduction of a maximally mixed state is again maximally mixed. We also have the following connection between $r$-uniform states and pure codes.

\begin{theorem}\label{thm:obs}
A subspace $\mathcal{Q}\subset (\CC^D)^{\otimes n}$ is a pure $((n,K,d))_D$ code with $d>1$ if and only if every state in $\mathcal{Q}$ is $(d-1)$-uniform.
\end{theorem}
\begin{proof}
    See \cite[{Observation 1}]{Huber2020}.
\end{proof}

An $r$-uniform state $\ket{\varphi}\in (\CC^D)^{\otimes n}$ with $r= \lfloor \frac{n}{2}\rfloor$ is called \textit{absolutely maximally entangled} (AME). While AME states have many applications in QIT \cite{PhysRevA.86.052335,Helwig:2013qoq,PhysRevA.104.032601}, we are primarily concerned with their relationship to QECCs. \Cref{thm:obs} gives one aspect of this relationship, but we will see more in the next section. If $n$ is even, the AME state $\ket{\varphi}$ is also called a \textit{perfect tensor} \cite{Pastawski2015}.

\section{The correspondence between perfect tensors and MDS codes}\label{sec:correspondence}

\subsection{New codes from old ones} \Cref{thm:newfromold} below provides a method for generating new pure codes from old ones. This was first formulated by Rains for qubit codes \cite{Rains1996QuantumWE}. The more general version for qudits appears in a paper by Huber and Grassl \cite{Huber2020}. We are interested in a special case of this result, which we state in \Cref{cor:codeconstruction}.

\begin{theorem}\label{thm:newfromold}
        Let $\Pi_\mathcal{Q}\in\text{End}((\CC^D)^{\otimes n})$ be the orthogonal projection onto a pure code $\mathcal{Q}$ with parameters $((n,K,d))_D$, where $n,d\geq 2$. The operator
    $
    D\cdot\text{Tr}_{\{1\}}(\Pi_\mathcal{Q})
    $
    is the projection onto a pure code with parameters $((n-1,DK,d-1))_D$.
\end{theorem}
\begin{proof}
    See the proofs of \cite[{Theorem 19}]{Rains1996QuantumWE} and \cite[{Theorem 2}]{Huber2020}. Note that the scale factor $D$ is needed to make the partial trace a projection. In general, the trace of a projection must be equal to the dimension of its image.
\end{proof}

\begin{corollary}\label{cor:codeconstruction}
    If $\ket{\varphi}\in (\CC^D)^{\otimes n}$ is a normalized $r$-uniform state, then  $D\cdot\text{Tr}_{\{1\}}(\ket{\varphi}\bra{\varphi})$ is the projection onto a pure $((n-1,D,r))_D$ code.
\end{corollary}
\begin{proof}
    $\ket{\varphi}\bra{\varphi}$ is the projection onto the subspace spanned by $\ket{\varphi}$, which, by \Cref{thm:obs}, is a pure $((n,1,r+1))_D$ code. So the result follows from \Cref{thm:newfromold}.
\end{proof}

Huber and Grassl show that the construction of \Cref{cor:codeconstruction} can be reversed when $\ket{\varphi}$ is a perfect tensor. Specifically, \Cref{thm:reverserains} holds. In \Cref{sec:lifting} we show that this gives a bijection of LU orbits.

\begin{theorem}[Huber-Grassl]\label{thm:reverserains}
    Let $\{\ket{\varphi_0},\dots,\ket{\varphi_{D-1}}\}$ be an orthogonal basis of a MDS code with parameters $((n-1,D,\frac{n}{2}))_D$, where $n$ is even. Then $\frac{1}{\sqrt{D}}\sum_i \ket{i}\ket{\varphi_i}\in (\CC^D)^{\otimes n}$ is a normalized AME state.
\end{theorem}
\begin{proof}
    See the proof of \cite[{Proposition 7}]{Huber2020}.
\end{proof}

\subsection{Correspondence of LU orbits}\label{sec:lifting}
Let $\mathcal{Q}$ be a $D$-dimensional subspace of $(\CC^D)^{\otimes (n-1)}$.
We construct a normalized state $\ket{\varphi_\mathcal{Q}}\in (\CC^D)^{\otimes n}$ as follows. Let $\Pi_\mathcal{Q}$ be the scaled projection onto $\mathcal{Q}$ defined by
\[
\Pi_\mathcal{Q}=\frac{1}{D}(\ket{\varphi_0}\bra{\varphi_0}+\dots+\ket{\varphi_{D-1}}\bra{\varphi_{D-1}}),
\]
where $\{\ket{\varphi_0},\dots,\ket{\varphi_{D-1}}\}$ is an orthogonal basis of $\mathcal{Q}$.
Regarding $\Pi_\mathcal{Q}$ as a mixed state, we purify to obtain $\ket{\varphi_\mathcal{Q}}=\frac{1}{\sqrt{D}}\sum_i \ket{i}\ket{\varphi_i}$. Although $\ket{\varphi_\mathcal{Q}}$ depends on the choice of orthogonal basis, its LU orbit is uniquely determined by $\mathcal{Q}$. We claim that the function
\begin{equation}\label{eq:function}
\text{LU orbit of $\mathcal{Q}$}\quad\mapsto\quad\text{LU orbit of $\ket{\varphi_{\mathcal{Q}}}$}
\end{equation}
whose domain is the collection of LU orbits of $D$-dimensional subspaces $\mathcal{Q}\subset(\CC^D)^{\otimes (n-1)}$ is well-defined. In other words, the following lemma holds:

\begin{lemma}
    $\mathcal{Q}$ and $\mathcal{Q'}$ are LU equivalent if and only if $\ket{\varphi_\mathcal{Q}}$ and $\ket{\varphi_\mathcal{Q'}}$ are LU equivalent.
\end{lemma}
\begin{proof}
    Observe that $\mathcal{Q}$ and $\mathcal{Q}'$ are equivalent if and only if $\Pi_\mathcal{Q}$ and $\Pi_{\mathcal{Q}'}$ are equivalent if and only if $\ket{\varphi_\mathcal{Q}}$ and $\ket{\varphi_\mathcal{Q'}}$ are equivalent.
\end{proof}

\begin{theorem}\label{thm:correspondence}
    Let $n\geq 2$ be even. There is a one-to-one correspondence between LU orbits of normalized AME states in $(\CC^D)^{\otimes n}$ and LU orbits of $((n-1,D,\frac{n}{2}))_D$ codes.
\end{theorem}
\begin{proof}
    Let $\mathcal{Q}$ be a $((n-1,D,\frac{n}{2}))_D$ code. Then, by \Cref{thm:reverserains}, $\ket{\varphi_\mathcal{Q}}$ is AME. This means that the restriction $\alpha$ of \eqref{eq:function} to LU orbits of $((n-1,D,\frac{n}{2}))_D$ codes maps into the collection of LU orbits of normalized AME states in $(\CC^D)^{\otimes n}$. 
    
    Next, we define a function $\beta$ from the collection of LU orbits of normalized AME states in $(\CC^D)^{\otimes n}$ to the collection of LU orbits of $((n-1,D,\frac{n}{2}))_D$ codes as follows. Given an orbit $\mathcal{O}$ in the domain of $\beta$, choose any representative $\ket{\varphi}\in \mathcal{O}$. By \Cref{cor:codeconstruction}, $D\cdot \text{Tr}_{\{1\}}(\ket{\varphi}\bra{\varphi})$ is the projection onto a pure $((n-1,D,\frac{n}{2}))_D$ code $\mathcal{Q}$. Define $\beta(\mathcal{O})$ to be the LU orbit of $\mathcal{Q}$.
    
    From the definitions of $\alpha$ and $\beta$, one checks that
    \begin{align*}
    \beta\circ\alpha(\mathcal{Q}) = \beta(\ket{\varphi_{\mathcal{Q}}})=\text{im }\text{Tr}_{\{1\}}(\ket{\varphi_{\mathcal{Q}}}\bra{\varphi_\mathcal{Q}})=\mathcal{Q},
    \end{align*}
    where we abuse notation by using orbit representatives instead of orbits. Similarly, letting $\ket{\varphi}=\frac{1}{\sqrt{D}}\sum_i \ket{i}\ket{\varphi_i}$, we have
    \begin{align*}
    \alpha\circ\beta(\ket{\varphi})&=\alpha(\text{im Tr}_{\{1\}}(\ket{\varphi}\bra{\varphi}))
        = \alpha(\text{im }\frac{1}{D}\sum_i\ket{\varphi_i}\bra{\varphi_i})
        =\ket{\varphi}.
    \end{align*}
     It follows that $\alpha$ and $\beta$ are bijections.
\end{proof}

\subsection{Transversal gates and local symmetries}\label{sec:gates&sym} Let $\ket{\varphi}\in (\CC^D)^{\otimes n}$ be an $r$-uniform state such that $\langle\varphi|\varphi\rangle=D$. There is a close relationship between the group
\[
S(\ket{\varphi})=\{g\in \GL_3^{\otimes 4}:g\ket{\varphi}=\ket{\varphi}\}
\]
of \textit{local symmetries} of $\ket{\varphi}$ and the transversal gates of the corresponding $((n-1,D,r))_D$ code $\mathcal{Q}$ coming from the construction of \Cref{cor:codeconstruction}. We aim to make this relationship precise in this section.

First, note that the code $\mathcal{Q}$ is has a natural orthogonal basis consisting of the elements
\begin{equation}\label{eq:basis}
\ket{\varphi_i}=\bra{i}_1\ket{\varphi}, \quad i=0,1,\dots,D-1.
\end{equation}
The notation used here refers to tensor contraction. That is, $\ket{\varphi_i}$ is the image of $\bra{i}$ under the linear map $V_1^*\to V_2\otimes \dots\otimes V_n$ corresponding to the tensor $\ket{\varphi}\in V_1\otimes V_2\otimes \dots\otimes V_n$, where each $V_k$ is an indexed copy of $\CC^D$. Thus, $\ket{\varphi}=\sum_i \ket{i}\ket{\varphi_i}$.
These $\ket{\varphi_i}$ span the image $\mathcal{Q}$ of the reduced density matrix $\text{Tr}_{\{1\}}(\ket{\varphi}\bra{\varphi})$, as we can see from
\[
\text{Tr}_{\{1\}}(\ket{\varphi}\bra{\varphi})=\text{Tr}_{\{1\}}\Big(\sum_{i,j=0}^{D-1} \ket{i}\bra{j}\otimes\ket{\varphi_i}\bra{\varphi_j}\Big)=\sum_{i=0}^{D-1} \ket{\varphi_i}\bra{\varphi_i}.
\]
Furthermore, since $\ket{\varphi}$ is $1$-uniform and $\langle\varphi|\varphi\rangle=D$,
\[
\text{Tr}_{{\{1\}}^C}(\ket{\varphi}\bra{\varphi})=\text{Tr}_{{\{1\}}^C}\Big(\sum_{i,j=0}^{D-1} \ket{i}\bra{j}\otimes\ket{\varphi_i}\bra{\varphi_j}\Big)=\langle\varphi_j|\varphi_i\rangle\ket{i}\bra{j}=\mathbb{I},
\]
where $\{1\}^C$ is the complement of $\{1\}$ in $\{1,2,\dots,n\}$. This shows that the basis $\{\ket{\varphi_i}\}$ is orthogonal.
\begin{theorem}\label{thm:localsym}
    Let $\ket{\varphi}\in(\CC^D)^{\otimes n}$ be an $r$-uniform state such that $\langle\varphi|\varphi\rangle=D$. Let $\mathcal{Q}$ be the subspace of $(\CC^D)^{\otimes (n-1)}$ spanned by the orthogonal basis vectors $\ket{\varphi_i}$ defined in \eqref{eq:basis}. Let $A_1,A_2,\dots,A_n\in \SL_D$ and $\lambda\in\CC\setminus \{0\}$. Then
    $\lambda(A_1\otimes A_2\otimes \dots\otimes A_n)\in S(\ket{\varphi})$ if and only if $A_2\otimes\dots\otimes A_n\in N(\mathcal{Q})$ and $\mu(A_2\otimes \dots\otimes A_n)=(\lambda A_1^\top)^{-1}.$ Here, as in \Cref{sec:transversal}, the representation $\mu:N(\mathcal{Q})\to\GL(\mathcal{Q})$ is defined by $\mu(g)_{ij}=\bra{\varphi_{i-1}}g\ket{\varphi_{j-1}}.$
\end{theorem}
\begin{proof}
    Let $V_i$ be a copy of $\CC^D$ for $i\in\{1,\dots,n\}$. As a tensor, $\ket{\varphi}\in V_1\otimes V_2\otimes \dots\otimes V_n$ corresponds to a linear map $V_1^*\to V_2\otimes \dots\otimes V_n$. Let $M$ be the $D^{n-1}\times D$ matrix associated to this linear map. We have
    \begin{align*}
        \lambda (A_1\otimes A_2\otimes\dots\otimes A_n)\in S(\ket{\varphi})&\iff\lambda (A_1\otimes A_2\otimes \dots\otimes A_n)\ket{\varphi}=\ket{\varphi}\\
        &\iff\lambda(A_2\otimes \dots \otimes A_n)MA_1^\top=M\\
        &\iff (A_2\otimes \dots\otimes A_n)M=M(\lambda A_1^\top)^{-1}.
    \end{align*}
    Let $\{f_1,\dots,f_D\}$ the standard basis of $V_1^*$ consisting of column vectors $f_j$. Then the above holds if and only if
    \begin{align*}
        (A_2\otimes \dots\otimes A_n)Mf_j=M(\lambda A_1^\top)^{-1}f_j
        = \lambda^{-1}M(\sum_{i=1}^D a_{ij}f_i)
        =\lambda^{-1}(\sum_{i=1}^D a_{ij}Mf_i)
    \end{align*}
    for every $f_j$, where the constants $a_{ij}$ are the entries of the matrix $(A_1^\top)^{-1}=(a_{ij})$. Since $Mf_j=\ket{\varphi_{j-1}}$, this is equivalent to $A_1\otimes\dots\otimes A_n\in N(\mathcal{Q})$ and $\mu(A_2\otimes\dots\otimes A_n)=(\lambda A_1^\top)^{-1}$.
\end{proof}
\Cref{thm:localsym} shows that $S(\ket{\varphi})$ and $N(\mathcal{Q})$ are in a certain sense the same object, as one determines the other. Furthermore, $S(\ket{\varphi})\leq \U_D^{\otimes n}$ if and only if $N(\mathcal{Q})\leq \SU_D^{\otimes {n-1}}$ and in this case there is a direct relationship between the local symmetry group of $\ket{\varphi}$ and the transversal gates of $\mathcal{Q}$. We shall see in \Cref{sec:mainresults} that this is the case when $\ket{\varphi}$ is a 4-qutrit AME state. In general, for any $r$-uniform state $\ket{\varphi}$, $S(\ket{\varphi})\leq \U_D^{\otimes n}$ whenever $S(\ket{\varphi})$ is finite \cite[Proposition 6]{Gour_2011}.

\subsection{Connections to prior literature}
Specializing to the case $n=4$, \Cref{thm:correspondence} tells us that there is a one-to-one correspondence between LU orbits of normalized AME states in $(\CC^D)^{\otimes 4}$ and LU orbits of $((3,D,2))_D$ codes. For $D=2$, this corresponds to the known facts that there does not exist a 4-qubit AME state  \cite{HIGUCHI2000213} and there does not exist a $((3,2,2))_2$ code; the latter is known by the linear programming method from the quantum MacWilliams identity \cite{Shor1997} (see \cite[\S~7.4.4]{gottesman2024qeccbook} for an explicit calculation). For $D\geq 3$, we have the following result of Rather et al.:

\begin{theorem}[Rather-Ramadas-Kodiyalam-Lakshminarayan, \cite{PhysRevA.108.032412}]\label{thm:Rather}
    There exists a unique normalized AME state in $(\CC^3)^{\otimes 4}$ up to LU equivalence. If $D>3$, then there are infinitely many LU orbits of normalized AME states in $(\CC^D)^{\otimes 4}$. 
\end{theorem}

\begin{corollary}\label{cor:Rather}
    The code $\mathcal{C}$ is the unique $((3,3,2))_3$ code up to LU equivalence.  If $D>3$, then there are infinitely many LU orbits of $((3,D,2))_D$ codes.
\end{corollary}
\begin{proof}
    This follows from \Cref{thm:correspondence,thm:Rather}.
\end{proof}

In this paper, we focus on the exceptional case $D=3$, where the AME state in $(\CC^D)^{\otimes 4}$ and the pure $((3,D,2))_D$ code exist and are unique.

Let us briefly mention what is known about $6$-qudit AME states. Rains showed that the normalized AME state of 6 qubits, the $((5,2,3))_2$ code, and the $((4,4,2))_2$ code are all unique up to LU equivalence \cite{746807}. In recent work, Ramadas and Lakshminarayan show that there exist infinitely many LU orbits of normalized AME states in $(\CC^D)^{\otimes 6}$ when $D\neq 2,3,6$ \cite{Ramadas_2025}. This implies that there exist infinitely many LU orbits of $((5,D,3))_D$ codes when $D\neq 2,3,6$. The cardinality of the set of LU orbits of normalized AME states in $(\CC^D)^{\otimes n}$ for the remaining cases $D=3,6$ is not known, although such AME states exist \cite{helwig2013absolutelymaximallyentangledqudit}.

\subsection{The 4-qutrit AME state and its associated code}\label{sec:332code}

Although the normalized 4-qutrit AME state is unique up to the action of the LU group, there are many known ways to construct a representative in its orbit \cites{helwig2013absolutelymaximallyentangledqudit,PhysRevA.72.012314,PhysRevA.92.032316,PhysRevA.91.012332,PhysRevResearch.3.043034}. We make use of a particularly simple version of the state that comes from a pair of orthogonal Latin squares (see \cites{PhysRevA.92.032316,PhysRevLett.128.080507} for an explanation and more details of the connection between orthogonal Latin squares and AME states). Specifically, the following is a 4-qutrit AME state:
\[
\ket{\Phi}=\frac{1}{\sqrt{3}}(\ket{0000}+ \ket{0111}+\ket{0222}+
\ket{1012}+\ket{1120}+\ket{1201}+
\ket{2021}+\ket{2102}+\ket{2210}).
\]
The state $\ket{\Phi}$ is not normalized but scaled so that the states $\ket{s_i}=\bra{i-1}_1\ket{\Phi}$ for $i\in\{1,2,3\}$ are normalized.
Let $\mathcal{C}$ denote the $((3,3,2))_3$ code arising from $\ket{\Phi}$ via the construction of \Cref{cor:codeconstruction}. As explained in \Cref{sec:gates&sym}, the vectors $\ket{s_i}$ form an orthogonal basis of $\mathcal{C}$. Written out, they are
\begin{align*}
\ket{s_1}&=\frac{1}{\sqrt{3}}(\ket{000}+\ket{111}+\ket{222}), \\
\ket{s_2}&=\frac{1}{\sqrt{3}}(\ket{012}+\ket{120}+\ket{201}), \\
\ket{s_3}&=\frac{1}{\sqrt{3}}(\ket{021}+\ket{102}+\ket{210}).
\end{align*}
One can check that $\mathcal{C}$ is the subspace consisting of states fixed by the group of operators generated by $X\otimes X\otimes X$ and $Z\otimes Z\otimes Z$, where $X$ and $Z$ are the $3\times 3$ generalized Pauli matrices of \eqref{eq:Pauli}. Thus, $\mathcal{C}$ is a stabilizer code. By a result of Hebenstreit et al. \cite{Hebenstreit}, there are no elements in $\SL_3^{\otimes 3}$ outside of the Pauli group $\mathcal{P}_3$ that also fix every point in $\mathcal{C}$.

\begin{theorem}[Hebenstreit-Spee-Kraus]\label{thm:kernel} Let $\mathcal{S}$ be the subgroup of $\SL_3^{\otimes 3}$ generated by $X^{\otimes 3}$ and $Z^{\otimes 3}$. Then
    \[\mathcal{S}=\{g\in \SL_3^{\otimes 3}:g\ket{\varphi}=\ket{\varphi},\:\forall\ket{\varphi}\in\mathcal{C}\}.\]
\end{theorem}
\begin{proof}
    We use \cite[Lemma 2]{Hebenstreit}, which states: if $\ket{\psi}\in\mathcal{C}$ is generic, that is, $\ket{\psi}$ is contained in the complement of a proper complex subvariety of $\mathcal{C}$, then 
    \[
    \mathcal{S}=\{g\in\SL_3^{\otimes 3}:g\ket{\psi}=\ket{\psi}\}.
    \]
    If $g\in\SL_3^{\otimes 3}$ such that $g\ket{\varphi}=\ket{\varphi}$ for all $\ket{\varphi}\in \mathcal{C}$, then $g\ket{\psi}=\ket{\psi}$ for generic $\ket{\psi}\in \mathcal{C}$. Hence $g\in \mathcal{S}$. Conversely, suppose that $g\in \mathcal{S}$. Then $(g-\mathbb{I})\ket{\psi}=0$ for generic $\ket{\psi}\in\mathcal{C}$. By continuity, $(g-\mathbb{I})\ket{\varphi}=0$ for all $\ket{\varphi}\in\mathcal{C}$.
\end{proof}

\section{Geometric invariant theory and applications}\label{sec:GIT}

\subsection{Critical points} In this section, we introduce critical points and the Kempf-Ness theorem. We state the Kempf-Ness theorem as it appears in \cite{Gour_2011}; this version is often enough for applications in QIT (e.g. \cites{4qubit,Gour_2011,10.1063/1.3511477,Burchardt}). For a more general version, one may refer to Wallach's book on GIT \cite{WallachGIT}. Define $\mathcal{H}$, $G$, and $K$ as in \Cref{sec:somegroups}.

\begin{definition}\label{def:criticalpoint}
    A state $\ket{\varphi}\in \mathcal{H}$ is \textit{critical} if either of the following equivalent conditions is met:
    \begin{itemize}
        \item[1.] $\bra{\varphi}E\ket{\varphi}=0$ for every $E$ in the Lie algebra of $G$.
        \item[2.] The reduction of $\ket{\varphi}\bra{\varphi}$ to any single-qudit subsystem is maximally mixed.
    \end{itemize}
    The equivalence of these two conditions is given by \cite[{Theorem 3}]{Gour_2011}.
\end{definition}

The first condition of \Cref{def:criticalpoint} is the standard definition of a critical point, while the second condition is more meaningful from the perspective of QIT. Note that if $d_1=\dots =d_n$, then the second condition states that $\ket{\varphi}$ is $1$-uniform. Thus, all $r$-uniform states are critical. Critical points are also called locally maximally entangled states in the literature \cites{Bryan2019locallymaximally,Slowik2021}. The existence of critical points and the dimension of the set of all critical points have been determined in all multipartite systems \cite{Bryan2018}; both depend on the local dimensions $(d_1,\dots,d_n)$.

\begin{theorem}[The Kempf-Ness theorem]\label{thm:Kempf-Ness}
Let $\ket{\varphi}\in \mathcal{H}$. The following are true:
\begin{itemize}
    \item[1.] $\ket{\varphi}$ is critical if and only if $\bra{\varphi}g^\dagger g \ket{\varphi}\geq \braket{\varphi|\varphi}$ for all $g\in G$.
    \item[2.] If $\ket{\varphi}$ is critical and $g\in G$ such that $\bra{\varphi}g^\dagger g \ket{\varphi}=\braket{\varphi|\varphi}$, then there exists $h\in K$ such that $g\ket{\varphi}=h\ket{\varphi}$.
    \item[3.] The $G$-orbit of $\ket{\varphi}$ is closed if and only if it contains a critical point.
\end{itemize}
\end{theorem}
\begin{proof}
    See \cite[{Theorem~2}]{Gour_2011}.
\end{proof}

The first item of \Cref{thm:Kempf-Ness} characterizes critical points in more geometrical terms: $\ket{\varphi}\in \mathcal{H}$ is critical if and only if it minimizes the norm in its $G$-orbit. The second item of \Cref{thm:Kempf-Ness} states that if two points of a $G$-orbit both minimize the norm, then the points are in fact related by a LU transformation. The third item of \Cref{thm:Kempf-Ness} characterizes closed $G$-orbits; such orbits are significant because they can be separated by $G$-invariant polynomials. Specifically, Gour and Wallach stated the following.

\begin{theorem}[Gour-Wallach, \cite{PhysRevLett.111.060502}]\label{thm:closed}
    Let $\ket{\varphi},\ket{\psi}\in\mathcal{H}$ be normalized states with closed $G$-orbits. There exists $\theta\in [0,2\pi)$ and $g\in G$ such that
    $\ket{\varphi}=e^{\im \theta} g\ket{\psi}/\sqrt{\bra{\varphi}g^\dagger g\ket{\psi}}$
    if and only if $f(\ket{\varphi})/h(\ket{\varphi})=f(\ket{\psi})/h(\ket{\psi})$ for all pairs of homogeneous $G$-invariant polynomials $f,h$ of the same degree with $h(\ket{\psi})\neq 0$.
\end{theorem}

By \Cref{thm:Kempf-Ness}, two critical states are $G$-equivalent if and only if they are $K$-equivalent. So $G$-invariant polynomials distinguish LU orbits of critical states.

\subsection{Vinberg theory}\label{sec:Vinberg} In this section, we state some interesting results from Vinberg's theory of graded Lie algebras \cite{Vinberg_1976}. The relevance to QIT should become clearer from the two examples of graded Lie algebras given in \Cref{sec:2graded,sec:3graded}. We begin by fixing some notation.

Let $\mathfrak{g}$ denote a $\ZZ_m$-graded complex Lie algebra for some $m\geq 2$. By definition, the fact that $\mathfrak{g}$ is $\ZZ_m$-graded means that there exist indexed subspaces $\mathfrak{g}_i$ such that \[\mathfrak{g}= \mathfrak{g}_0\oplus\mathfrak{g}_1\oplus\dots\oplus\mathfrak{g}_{m-1}\quad \text{and}\quad[\mathfrak{g}_i,\mathfrak{g}_j]\subset\mathfrak{g}_{i+j \mod m}.\] Let $\tilde{G}$ be the connected Lie group corresponding to $\mathfrak{g}$ and let $\tilde{G}_0\leq \tilde{G}$ be the connected subgroup corresponding to $\mathfrak{g}_0$. Since $[\mathfrak{g}_0,\mathfrak{g}_1]\subset\mathfrak{g}_1$, the adjoint action of $\mathfrak{g}_0$ restricts to an algebra representation
\[
\mathfrak{g}_0\to\text{End}(\mathfrak{g}_1),\quad E\mapsto \text{ad}(E)|_{\mathfrak{g}_1}
\]
which corresponds to the group representation
\[\pi:\tilde{G}_0\to\GL(\mathfrak{g}_1),\quad A\mapsto\text{Ad}(A)|_{\mathfrak{g}_1}.\]
Let $L=\pi(\tilde{G}_0)$.

\begin{definition}\label{defi:closedorbit}
    An element $x\in \mathfrak{g}_1$ is \textit{semisimple} if its orbit under the action of $L$ is closed (see \cite[{Proposition 3}]{Vinberg_1976}).
    A subspace $\mathfrak{c}\subset\mathfrak{g}_1$ is a \textit{Cartan subspace} if it is maximal with respect to the following properties: $[\mathfrak{c},\mathfrak{c}]=0$ and every element of $\mathfrak{c}$ is semisimple.
\end{definition}

\begin{theorem}\label{thm:semisimpleclassification}
    For any Cartan subspaces $\mathfrak{c}$ and $\mathfrak{c}'$ of $\mathfrak{g}_1$, there exists an element $A\in L$ such that $A\mathfrak{c}=\mathfrak{c}'$. Consequently, every semisimple element of $\mathfrak{g}_1$ is mapped by some element of $L$ to a fixed Cartan subspace.
\end{theorem}
\begin{proof}
    See \cite[{Theorem 1}]{Vinberg_1976}.
\end{proof}

\begin{definition}
The \textit{normalizer} of a Cartan subspace $\mathfrak{c}\subset\mathfrak{g}_1$ is the subgroup of $L$ consisting of elements that fix $\mathfrak{c}$, denoted
$
N(\mathfrak{c})=\{A\in L:A\mathfrak{c}=\mathfrak{c}\}.
$
Let $\mu$ denote the obvious representation $N(\mathfrak{c})\to\GL(\mathfrak{c})$. The kernel of $\mu$ is the \textit{centralizer} of $\mathfrak{c}$, denoted
$
C(\mathfrak{c})=\{A\in L:Ax=x\text{ for all }x\in\mathfrak{c}\}.
$
The \textit{Weyl group} $W(\mathfrak{c})$ is the image $\mu(N(\mathfrak{c}))\cong N(\mathfrak{c})/C(\mathfrak{c}).$
\end{definition}

\begin{theorem}\label{thm:weylgroupequivalence}
    If the element $x$ of the Cartan subspace $\mathfrak{c}$ can be mapped to an element $y$ of $\mathfrak{c}$ by a transformation from $L$, then $x$ can be mapped to $y$ by a transformation from $W(\mathfrak{c})$.
\end{theorem}
\begin{proof}
    See \cite[{Theorem 2}]{Vinberg_1976}.
\end{proof}

If $\mathfrak{c}$ and $\mathfrak{c}'$ are Cartan subspaces of $\mathfrak{g}_1$, then, by \Cref{thm:semisimpleclassification}, $N(\mathfrak{c})$ and $N(\mathfrak{c}')$ are isomorphic by a similarity transformation. Thus, despite our choice of notation, the Weyl group $W(\mathfrak{c})$ is determined up to isomorphism by $\mathfrak{g}$ and its $\ZZ_m$-grading.

\begin{theorem}\label{thm:restriction}
    The restriction map $r:\CC[\mathfrak{g}_1]^L\to \CC[\mathfrak{c}]^{W(\mathfrak{c})}$ where $r(f)= f|_{\mathfrak{c}}$ is an isomorphism of invariant algebras.
\end{theorem}
\begin{proof}
    See \cite[{Theorem 7}]{Vinberg_1976}.
\end{proof}

\begin{definition}
    An element of $\GL(\CC^K)$ is a \textit{complex reflection} if it is diagonalizable, has one eigenvalue equal to a root of unity, and all but one of its eigenvalues is equal to 1. A subgroup of $\GL(\CC^K)$ is a \textit{complex reflection group} if it is generated by complex reflections.
\end{definition}

\begin{theorem}\label{thm:complexreflection}
    The Weyl group of a $\ZZ_m$-graded Lie algebra is a finite complex reflection group.
\end{theorem}
\begin{proof}
    See \cite[{Proposition 8}]{Vinberg_1976} and \cite[{Theorem 8}]{Vinberg_1976}.
\end{proof}

\subsection{Embedding into an exceptional Lie algebra}\label{sec:3graded} Let $\mathcal{H}_{333}=(\CC^3)^{\otimes 3}$. In what follows, we describe a construction of the exceptional Lie algebras $\mathfrak{e}_8$ and $\mathfrak{e}_6$ given by Vinberg and Elashvili \cite{Vinberg-Elasvili}. This is of interest to us because the $\SL_3^{\otimes 3}$-module $\mathcal{H}_{333}$ is naturally embedded in $\mathfrak{e}_6$ through this construction.

According to \cite{Vinberg-Elasvili}, $\mathfrak{e}_8\cong \mathfrak{g}'_0\oplus\mathfrak{g}'_1\oplus\mathfrak{g}'_2$ admits a decomposition as a $\ZZ_3$-graded algebra with the following components:
\[
\mathfrak{g}'_0\cong \mathfrak{sl}_9,\quad \mathfrak{g}'_1\cong \Lambda^3\CC^9,\quad \mathfrak{g}'_2\cong \Lambda^3(\CC^9)^*.
\]
The Lie bracket restricted to $\mathfrak{g}'_0\times\mathfrak{g}'_1$ is the usual action of $\mathfrak{sl}_9$ on $\Lambda^3 \CC^9$. That is,
\begin{equation}\label{eq:adjoint}
[E,v_1\wedge v_2\wedge v_3]=Ev_1\wedge v_2\wedge v_3+v_1\wedge Ev_2\wedge v_3+v_1\wedge v_2\wedge Ev_3
\end{equation}
for all $E\in\mathfrak{sl}_9$ and $v_i\in \CC^9$.

Next, we describe $\mathfrak{e}_6$ as a subspace of $\mathfrak{e}_8$ as in \cite{Vinberg-Elasvili}. We may decompose $\CC^9\cong V_1\oplus V_2\oplus V_3$ into a direct sum of 3-dimensional subspaces $V_i$. To choose coordinates, let $\{e_1,\dots,e_9\}$ be a basis of $\CC^9$ and let $V_1=\text{span}\{e_1,e_2,e_3\}$, $V_2=\text{span}\{e_4,e_5,e_6\}$, $V_3=\text{span}\{e_7,e_8,e_9\}.$ We have the injection
\begin{equation}\label{eq:grade1}
V_1\otimes V_2\otimes V_3\to \Lambda^3 \CC^9,\quad v_1\otimes v_2\otimes v_3\mapsto v_1\wedge v_2\wedge v_3
\end{equation}
into the grade-1 subspace
and a similar injection $V_1^*\otimes V_2^*\otimes V_3^*\to\Lambda^3(\CC^9)^*$ into the grade-2 subspace. We also have the injection
\begin{equation}\label{eq:grade0}
\mathfrak{sl}_3\oplus\mathfrak{sl}_3\oplus\mathfrak{sl}_3\to\mathfrak{sl}_9,\quad
(E_1,E_2,E_3)\mapsto \begin{pmatrix}
    E_1 & 0 & 0 \\
    0 & E_2 & 0 \\
    0 & 0 & E_3
\end{pmatrix}
\end{equation}
obtained by sending triples to block diagonal matrices. The union of the images of these injections spans a subalgebra of $\mathfrak{e}_8$ which is isomorphic to $\mathfrak{e}_6$. Thus, we have established the $\ZZ_3$-grading \eqref{eq:z3grading} of $\mathfrak{e}_6$. From the adjoint action \eqref{eq:adjoint} and the injections \eqref{eq:grade1} and \eqref{eq:grade0}, we have
\[
[(E_1,E_2,E_3),v_1\otimes v_2\otimes v_3]=E_1v_1\otimes v_2\otimes v_3+v_1\otimes E_2v_2\otimes v_3+v_1\otimes v_2\otimes E_3v_3
\]
for $E_i\in\mathfrak{sl}_3$ and $v_i\in \CC^3$. That is, the adjoint action
\[
\mathfrak{sl}_3^{\times 3}\to\text{End}(\mathcal{H}_{333}),\quad E\mapsto \text{ad}(E)|_{\mathcal{H}_{333}}
\]
is the usual Lie algebra action. Thus, the corresponding group representation
\[
\pi:\SL_3^{\times 3}\to\GL(\mathcal{H}_{333}),\quad A\mapsto \text{Ad}(A)|_{\mathcal{H}_{333}}
\]
is the expected one. That is, $\pi(A_1,A_2,A_3)=A_1\otimes A_2\otimes A_3$ where $A_i\in\SL_3$.

We showed that the $\SL_3^{\otimes 3}$-module $\mathcal{H}_{333}$ is isomorphic to the grade-1 subspace of the $\ZZ_3$-graded Lie algebra $\mathfrak{e}_6$. This allows us to make use of Vinberg's results in the 3-qutrit setting. To be clear, one applies the theorems of \Cref{sec:Vinberg} by making the substitutions $\mathfrak{g}_1=\mathcal{H}_{333}$, $L=\SL_3^{\otimes 3}$, and $\mathfrak{c}=\mathcal{C}$. Note also the following:
\begin{lemma}\label{lem:semisimplecritical}
    A state $\ket{\varphi}\in \mathcal{H}_{333}$ is semisimple if it is critical.
\end{lemma}
\begin{proof}
    This follows from \Cref{thm:Kempf-Ness} and \Cref{defi:closedorbit}.
\end{proof}

\subsection{Cartan subspace of \texorpdfstring{$\mathcal{H}_{333}$}{H333}} Nurmiev showed that the span of \[\{\ket{s_1},\ket{s_2},\ket{s_3}\}\cup\{\bra{s_1},\bra{s_2},\bra{s_3}\}\] in $\mathcal{H}_{333}\oplus\mathcal{H}^*_{333}$ is a Cartan subalgebra of $\mathfrak{e}_6$ and $\mathcal{C}=\text{span}\{\ket{s_1},\ket{s_2},\ket{s_3}\}$ is a Cartan subspace \cite{Nurmiev_2000}. Thus, $\mathcal{C}$ is simultaneously the unique $((3,3,2))_3$ code up to $\SU_3^{\otimes 3}$-equivalence (\Cref{cor:Rather}) and the unique Cartan subspace of $\mathcal{H}_{333}$ up to $\SL_3^{\otimes 3}$-equivalence (\Cref{thm:semisimpleclassification}). We also have the following classification of AME states, which is analogous to \Cref{thm:semisimpleclassification,thm:weylgroupequivalence}. Note that a 3-qutrit state is AME if and only if it is critical.

\begin{theorem}[Classification of 3-qutrit AME states]\label{thm:AMEclassification}
    Every AME state in $\mathcal{H}_{333}$ is $\SU_3^{\otimes 3}$-equivalent to a point in the unique $((3,3,2))_3$ code $\mathcal{C}$. Two points in $\mathcal{C}$ are $\SU_3^{\otimes 3}$-equivalent if and only if there is a transversal gate mapping one to the other.
\end{theorem}
\begin{proof}
    By \Cref{thm:QMDS-pure}, since $\mathcal{C}$ is MDS, it is pure. Then, by \Cref{thm:obs}, $\mathcal{C}$ consists of critical points. Suppose $\ket{\varphi}\in\mathcal{H}_{333}$ is AME. Then, by \Cref{lem:semisimplecritical}, $\ket{\varphi}$ is semisimple and, by \Cref{thm:semisimpleclassification}, $\ket{\varphi}$ is $\SL_{3}^{\otimes 3}$-equivalent to a point in $\mathcal{C}$ which must be critical. Thus, by \Cref{thm:Kempf-Ness}, $\ket{\varphi}$ is $\SU_3^{\otimes 3}$-equivalent to a point in $\mathcal{C}$.

    Suppose two points $\ket{\varphi},\ket{\psi}\in\mathcal{C}$ are $\SU_3^{\otimes 3}$-equivalent. Since $\SU_3^{\otimes 3}\leq \SL_3^{\otimes 3}$, we may apply \Cref{thm:weylgroupequivalence} to conclude that there is an element of $W(\mathcal{C})$ mapping $\ket{\varphi}$ to $\ket{\psi}$. We will see in \Cref{sec:transversalgates} that the Weyl group $W(\mathcal{C})$ is the group of transversal gates on $\mathcal{C}$.
\end{proof}

\subsection{Four qubits}\label{sec:2graded} Another graded Lie algebra of interest in QIT is the $\ZZ_2$-graded Lie algebra
$
\mathfrak{so}_8\cong \mathfrak{sl}_2^{\times 4}\oplus (\CC^2)^{\otimes 4}.
$
The Lie bracket $\mathfrak{so}_8\times\mathfrak{so}_8\to\CC$ restricted to $\mathfrak{sl}_2^{\times 4}\times (\CC^2)^{\otimes 4}$ is defined by the usual Lie algebra action of $\mathfrak{sl}_2^{\times 4}$ on $(\CC^2)^{\otimes 4}$. Thus the 4-qubit state space $(\CC^2)^{\otimes 4}$ is isomorphic to the grade-1 subspace of $\mathfrak{so}_8$ as a $\SL_2^{\otimes 4}$-module. As in the 3-qutrit case, there is a Cartan subspace of $(\CC^2)^{\otimes 4}$ which is simultaneously a stabilizer code: namely, the subspace consisting of all points fixed by the group generated by $X\otimes X\otimes X\otimes X$ and $Z\otimes Z\otimes Z\otimes Z$. This subspace is known to be the unique $((4,4,2))_2$ code up to LU equivalence \cite{746807}. The author intends to investigate this situation more in future work.

\subsection{Invariants of 3 qutrits}\label{sec:invariantalgebra} In light of \Cref{thm:closed}, it would be useful to know all of the $\SL_D^{\otimes n}$-invariant polynomials on the space $(\CC^D)^{\otimes n}$. Although there is a known algorithm for computing these invariants up to any degree \cite{PhysRevLett.111.060502}, it quickly becomes intractable. For example, we do not have generators of the invariant ring for more than 4 qubits. See \cite{Luque_2006} for what is known about the invariants of 5 qubits. However, we know that the invariant algebra cannot be too complicated in the Vinberg setting because $\CC[\mathfrak{g}_1]^L$ is isomorphic to the algebra $\CC[\mathfrak{c}]^{W(\mathfrak{c})}$ of invariants of a finite group. In addition to being finite, the Weyl group $W(\mathfrak{c})$ is a complex reflection group, so $\CC[\mathfrak{c}]^{W(\mathfrak{c})}$ and $\CC[\mathfrak{g}_1]^L$ are both free (see the corollary of Theorem 8 in \cite{Vinberg_1976}).

Consider the $\SL_3^{\otimes 3}$-module $\mathfrak{g}_1=\mathcal{H}_{333}$. Here, the invariant algebra is known to be generated by three homogeneous polynomials of degrees $6$, $9$, and $12$ \cite[{\S~9}]{Vinberg_1976}. This implies that $W(\mathcal{C})$ is a group of order $6\times 9\times 12 = 648$ \cite[\S~1.2]{Cohen_1976}. Any element of $\mathcal{C}$ can be written as a linear combination $a\ket{s_1}+b\ket{s_2}+c\ket{s_3}$ of our chosen basis vectors. This allows us to write each element of the invariant algebra $\CC[\mathcal{C}]^{W(\mathcal{C})}\cong \CC[a,b,c]^{W(\mathcal{C})}$ in terms of the indeterminates $a,b,c$. Bremner et al. compute the fundamental $\SL_3^{\otimes 3}$-invariants \cite{Bremner2013FundamentalIF} then restrict to the Cartan subspace to find the following expressions \cites{BremnerHuOeding2014,Jaffali_2024} for the generators of $\CC[a,b,c]^{W(\mathcal{C})}$:
\begin{align*}
    I_6&=a^6+b^6+c^6-10(a^3 b^3 + a^3 c^3 + b^3 c^3),\\
    I_9&= (a^3-b^3)(a^3-c^3)(b^3-c^3),
\end{align*}
and
\begin{multline*}
    I_{12}=a^9(b^3+c^3)+b^9(a^3+c^3)+c^9(a^3+b^3) \\ -4(a^6 b^6+a^6 c^6 + b^6 c^6)+2(a^6 b^3 c^3+a^3 b^6 c^3+a^3 b^3 c^6).
\end{multline*}

\section{Computation of group generators}\label{sec:mainresults}

\subsection{Transversal gates of the \texorpdfstring{$((3,3,2))_3$}{((3,3,2))\_3} code}\label{sec:transversalgates} In this section, we find generators of the Weyl group $W(\mathcal{C})$ together with elements of the normalizer $N(\mathcal{C})$ that map to each generator under the representation $\mu:N(\mathcal{C})\to \GL(\mathcal{C})$ defined in \eqref{eq:mu}.

By \Cref{thm:complexreflection}, $W(\mathcal{C})$ is a finite complex reflection group. The finite complex reflection groups were classified by Shephard and Todd in 1954 \cite{Shephard_Todd_1954}. In \cite{Vinberg_1976}, Vinberg identified the Weyl group $W(\mathcal{C})$ associated to the $\ZZ_3$-graded Lie algebra $\mathfrak{e}_6$ as No. 25 in the Shephard and Todd's list of irreducible finite complex reflection groups \cite[{Table VII}]{Shephard_Todd_1954}. Cohen gives detailed instructions on how to compute the generators of this group \cite{Cohen_1976}, which we proceed to follow.

Given $a\in \CC^3$ and $d\in \mathbb{N}=\{1,2,\dots\}$, there is a complex reflection $s_{a,d}:\CC^3\to\CC^3$ that fixes the subspace of vectors orthogonal to $a$ and maps $a$ to $e^{2\pi \im/d}a$, given in \cite[\S~1.6]{Cohen_1976} by the formula
\begin{equation}\label{eq:reflection}
s_{a,d}(x)=x-(1-e^{2\pi\im/d})\frac{a^\dagger x}{a^\dagger a}a,\quad\forall x\in \CC^3.
\end{equation}

\begin{figure}
    \centering
    \begin{tikzpicture}[every node/.style={circle, draw, fill=white, inner sep=2pt}]
  \node (v1) at (0,0) [label=below:$e_1$] {$3$};
  \node (v2) at (3,0) [label=below:$e_2$] {$3$};
  \node (v3) at (6,0) [label=below:$e_3$] {$3$};

  \draw (v1) -- node[draw=none, fill=none, above] {} (v2)
        -- node[draw=none, fill=none, above] {} (v3);
\end{tikzpicture}
    \caption{A visualization of the vector graph $L_3$ associated to $W(\mathcal{C})$. The vectors $e_i$ are defined in \eqref{eq:vectors}. There is no edge between two vectors if and only if they are orthogonal.}
    \label{fig:placeholder}
\end{figure}

A \textit{vector graph} is a pair $\Gamma=(B,w)$ consisting of a nonempty finite collection of vectors $B\subset \CC^K$ (satisfying an additional property we do not care about) and a map $w:B\to \mathbb{N}\setminus\{1\}$ \cite[{\S~4.1}]{Cohen_1976}. Associated to a vector graph is the group $W(\Gamma)$ generated by all reflections $s_{a,w(a)}$ such that $a\in B$ \cite[{\S~4.2}]{Cohen_1976}. Set $K=3$ and consider the vector graph $L_3=(B,w)$ where $B$ consists of three unit vectors
\begin{equation}\label{eq:vectors}
e_1 = \ket{2},\quad e_2=\frac{\im}{\sqrt{3}}(\ket{0}+\ket{1}+\ket{2}),\quad e_3=\ket{1}
\end{equation}
and $w(e_i)=3$ for all $e_i\in B$. The group $W(L_3)$ is the Weyl group $W(\mathcal{C})$ \cite[\S~4.16]{Cohen_1976}. Thus, using \eqref{eq:reflection}, we find that $W(\mathcal{C})$ is generated by the matrices
$
R_i = \mathbb{I}-(1-\omega)e_i e_i^\dagger
$ where $\omega=e^{2\pi \im/3}$. Written out explicitly, these matrices are
\[
R_1=
\begin{pmatrix}
    1 & 0 & 0 \\
    0 & 1 & 0 \\
    0 & 0 & \omega
\end{pmatrix},\quad
R_2=
\frac{1}{\sqrt{3}}e^{\pi \im/6}
\begin{pmatrix}
    1 & \omega & \omega \\
    \omega & 1 & \omega \\
    \omega & \omega & 1
\end{pmatrix},\quad
R_3=
\begin{pmatrix}
    1 & 0 & 0 \\
    0 & \omega & 0 \\
    0 & 0 & 1
\end{pmatrix}.
\]
This representation of the Weyl group is in agreement with our basis $\{\ket{s_1},\ket{s_2},\ket{s_3}\}$, which can be shown by checking that the matrices above fix the invariants $I_6$, $I_9$, and $I_{12}$ from \Cref{sec:invariantalgebra}. With this knowledge, we find that the following elements of $\SU_3^{\otimes 3}$ map to each generator $R_i$ under $\mu:N(\mathcal{C})\to \GL(\mathcal{C})$.
\begin{align}\label{eq:cosetrepresentatives}
\begin{split}
\omega
\begin{pmatrix}
    1 & 0 & 0 \\
    0 & 1 & 0 \\
    0 & 0 & \omega^2
\end{pmatrix}
\otimes
\begin{pmatrix}
    1 & 0 & 0 \\
    0 & \omega^2 & 0 \\
    0 & 0 & 1
\end{pmatrix}
\otimes
\begin{pmatrix}
    \omega^2 & 0 & 0 \\
    0 & 1 & 0 \\
    0 & 0 & 1
\end{pmatrix}
&\mapsto R_1\\
\frac{1}{\sqrt{27}}e^{\pi \im/6}
\begin{pmatrix}
    \omega & 1 & 1 \\
    1 & \omega & 1 \\
    1 & 1 & \omega
\end{pmatrix}
\otimes\begin{pmatrix}
    \omega & 1 & 1 \\
    1 & \omega & 1 \\
    1 & 1 & \omega
\end{pmatrix}
\otimes \begin{pmatrix}
    \omega & 1 & 1 \\
    1 & \omega & 1 \\
    1 & 1 & \omega
\end{pmatrix}
&\mapsto R_2\\
\omega
\begin{pmatrix}
    \omega^2 & 0 & 0 \\
    0 & 1 & 0 \\
    0 & 0 & 1
\end{pmatrix}
\otimes
\begin{pmatrix}
    1 & 0 & 0 \\
    0 & \omega^2 & 0 \\
    0 & 0 & 1
\end{pmatrix}
\otimes
\begin{pmatrix}
    1 & 0 & 0 \\
    0 & 1 & 0 \\
    0 & 0 & \omega^2
\end{pmatrix}
&\mapsto R_3
\end{split}
 \end{align}
Let us briefly explain the process by which the operators on the left side of \eqref{eq:cosetrepresentatives} were found. For each $i$, we want some $Q_i\in\SL_3^{\otimes 3}$ such that $\mu(Q_i)=R_i$. Guessing that such $Q_i$ exist in $\SU_3^{\otimes 3}$, for each $i$ we used a numerical optimization algorithm to find $E_1,E_2,E_3$ in the Lie algebra of $\SU_3$ such that
\[
\|\mu'(\exp(E_1)\otimes \exp(E_2)\otimes \exp(E_3))-R_i\|=0,
\]
where the image of $\mu'$ is a matrix with entries defined by $\mu'(U)_{ij}=\bra{s_i}U\ket{s_j}.$ After this, we recognized the exact solutions and verified with a computer algebra system that \eqref{eq:cosetrepresentatives} holds. This computation gives the following result.

\begin{theorem}
    The Weyl group $W(\mathcal{C})$ is the group of transversal gates on $\mathcal{C}$.
\end{theorem}
\begin{proof}
Recall that $W(\mathcal{C})=\mu(N(\mathcal{C}))$ by definition. Thus, $N(\mathcal{C})$ is generated by the kernel $C(\mathcal{C})$ of $\mu$ together with any three representatives $Q_i\in \mu^{-1}(R_i)$ in the preimages of the generators of $W(\mathcal{C})$. By \Cref{thm:kernel}, $C(\mathcal{C})$ is a subgroup of $\SU_3^{\otimes 3}$. By \eqref{eq:cosetrepresentatives}, each $Q_i\in\SU_3^{\otimes 3}$. It follows that $N(\mathcal{C})$ is contained in $\SU_3^{\otimes 3}$.
\end{proof}

We note that a different set of generators for $W(\mathcal{C})$ is given in \cite{10.1063/5.0156805}. However, the authors of this previous work did not recognize elements of the Weyl group as transversal gates, nor did they describe the elements of $N(\mathcal{C})$, which are needed to implement the transversal gates.

\subsection{Local symmetries of the 4-qutrit AME state} The results of \Cref{sec:gates&sym} give us a connection between $N(\mathcal{C})$, $W(\mathcal{C})$, and the group of local symmetries $S(\ket{\Phi})$. Since we already understand $N(\mathcal{C})$ and $W(\mathcal{C})$, we can thus find generators of $S(\ket{\Phi})$.

\begin{corollary}\label{cor:localsym} The group $S(\ket{\Phi})$ is given by
    \[S(\ket{\Phi})=\{\overline{\mu(A_1\otimes A_2\otimes A_3)}\otimes A_1\otimes A_2\otimes A_3: A_1\otimes A_2\otimes A_3\in N(\mathcal{C})\},\]
    where the bar over the matrix $\mu(A_1\otimes A_2\otimes A_3)$ denotes entrywise complex conjugation.
\end{corollary}
\begin{proof}
    Notice the following facts and apply \Cref{thm:localsym}. Every element of $\GL_3^{\otimes 4}$ has the form $(\lambda A_0)\otimes A_1\otimes A_2\otimes A_3$, where $\lambda\in \CC\setminus\{0\}$ and each $A_i\in\SL_3$. Every element in the image of $\mu$ is unitary and taking the inverse transpose of a unitary matrix has the effect of conjugating every entry of the matrix.
\end{proof}

\begin{theorem}
    The group $S(\ket{\Phi})$ is of order 5,832 and is generated by the five LU operators displayed in \Cref{fig:localsymmetrygroup}.
\end{theorem}
\begin{proof}
    By \Cref{cor:localsym}, $S(\ket{\Phi})$ consists of all operators of the form $\bar{U}_0\otimes U_1\otimes U_2\otimes U_3$ where $U_0\in W(\mathcal{C})$ and $U_1\otimes U_2\otimes U_3\in \mu^{-1}(U_0)$. Let $Q_i\in\mu^{-1}(R_i)$ for $i\in\{1,2,3\}$ be the three coset representatives from \eqref{eq:cosetrepresentatives}. Every element of $\mu^{-1}(R_i)$ has the form $PQ_i$ for some $P\in C(\mathcal{C})$ in the kernel of $\mu$. Since $C(\mathcal{C})$ is generated by $X^{\otimes 3}$ and $Z^{\otimes 3}$ (\Cref{thm:kernel}) and the Weyl group $W(\mathcal{C})$ is generated by $R_1$, $R_2$, $R_3$, we conclude that $S(\ket{\Phi})$ is generated by $\mathbb{I}\otimes X^{\otimes 3}$, $\mathbb{I}\otimes Z^{\otimes 3}$, and $\bar{R}_i\otimes Q_i$ for $i=1,2,3$. These are the matrices in \Cref{fig:localsymmetrygroup}.
    
    Recall from \Cref{sec:invariantalgebra} that $W(\mathcal{C})$ is a group of order $648$. We determine the order of $S(\ket{\Phi})$ from the calculation $|W(\mathcal{C})|\times |C(\mathcal{C})|=648\times 9=5,832.$
\end{proof}

\begin{figure}
    \centering
    \[
\begin{pmatrix}
    1 & 0 & 0 \\
    0 & 1 & 0 \\
    0 & 0 & 1
\end{pmatrix}
\otimes
\begin{pmatrix} 0 & 0 & 1\\
1 & 0 & 0\\
0 & 1 & 0
\end{pmatrix}
\otimes
\begin{pmatrix} 0 & 0 & 1\\
1 & 0 & 0\\
0 & 1 & 0
\end{pmatrix}
\otimes
\begin{pmatrix} 0 & 0 & 1\\
1 & 0 & 0\\
0 & 1 & 0
\end{pmatrix}
\]
    \[
\begin{pmatrix}
    1 & 0 & 0 \\
    0 & 1 & 0 \\
    0 & 0 & 1
\end{pmatrix}
\otimes
\begin{pmatrix}
    1 & 0 & 0 \\
    0 & \omega & 0\\
    0 & 0 & \omega^2
\end{pmatrix}
\otimes
\begin{pmatrix}
    1 & 0 & 0 \\
    0 & \omega & 0\\
    0 & 0 & \omega^2
\end{pmatrix}
\otimes
\begin{pmatrix}
    1 & 0 & 0 \\
    0 & \omega & 0\\
    0 & 0 & \omega^2
\end{pmatrix}
\]
\[
\omega
\begin{pmatrix}
    1 & 0 & 0 \\
    0 & 1 & 0 \\
    0 & 0 & \omega^2
\end{pmatrix}
\otimes
\begin{pmatrix}
    1 & 0 & 0 \\
    0 & 1 & 0 \\
    0 & 0 & \omega^2
\end{pmatrix}
\otimes
\begin{pmatrix}
    1 & 0 & 0 \\
    0 & \omega^2 & 0 \\
    0 & 0 & 1
\end{pmatrix}
\otimes
\begin{pmatrix}
    \omega^2 & 0 & 0 \\
    0 & 1 & 0 \\
    0 & 0 & 1
\end{pmatrix}
\]
\[
\frac{\omega^2}{9}
\begin{pmatrix}
    \omega & 1 & 1 \\
    1 & \omega & 1 \\
    1 & 1 & \omega
\end{pmatrix}
\otimes
\begin{pmatrix}
    \omega & 1 & 1 \\
    1 & \omega & 1 \\
    1 & 1 & \omega
\end{pmatrix}
\otimes
\begin{pmatrix}
    \omega & 1 & 1 \\
    1 & \omega & 1 \\
    1 & 1 & \omega
\end{pmatrix}
\otimes
\begin{pmatrix}
    \omega & 1 & 1 \\
    1 & \omega & 1 \\
    1 & 1 & \omega
\end{pmatrix}
\]
\[
\omega
\begin{pmatrix}
    1 & 0 & 0 \\
    0 & \omega^2 & 0 \\
    0 & 0 & 1
\end{pmatrix}
\otimes
\begin{pmatrix}
    \omega^2 & 0 & 0 \\
    0 & 1 & 0 \\
    0 & 0 & 1
\end{pmatrix}
\otimes
\begin{pmatrix}
    1 & 0 & 0 \\
    0 & \omega^2 & 0 \\
    0 & 0 & 1
\end{pmatrix}
\otimes
\begin{pmatrix}
    1 & 0 & 0 \\
    0 & 1 & 0 \\
    0 & 0 & \omega^2
\end{pmatrix}
\]
    \caption{Generators of the local symmetry group $S(\ket{\Phi})$. We let $\omega=e^{2\pi \im/3}$.}
    \label{fig:localsymmetrygroup}
\end{figure}

\section{Conclusions and outlook}
We proved that a certain construction gives a bijection between LU orbits of normalized perfect tensors in $(\CC^D)^{\otimes n}$ and LU orbits of $((n-1,D,n/2))_D$ MDS codes, extending prior work of Rains, Huber, and Grassl \cite{Rains1996QuantumWE, Huber2020}. Furthermore, the local symmetries of an $r$-uniform state tell us about the transversal gates of its corresponding MDS code; this adds to the practical meaning of a state's local symmetries, which is known to give information about a state's reachability \cite{PhysRevLett.118.040503}. On the other hand, the transversal gates of a $((n-1,D,\frac{n}{2}))_D$ code do not determine the local symmetries of its corresponding perfect tensor unless the local symmetries are all LU operators. A possible avenue for future research is to consider more general situations where perfect tensors only have LU local symmetries. In the present work, we focus on the 4-qutrit AME state and the $((3,3,2))_3$ code. This is motivated by a connection to the graded Lie algebra $\mathfrak{e}_6$, which allows for the application of Vinberg's theory. As mentioned in \Cref{sec:2graded}, the state space of four qubits can also be embedded in a graded Lie algebra; the author intends to study further the consequences of this in future work.

\section{Acknowledgements}
I thank Eric Kubischta for his comments on the first draft of this paper.

\section{Declarations}

\noindent\textbf{Funding declaration.} This research received no external funding.
\bigskip

\noindent\textbf{Ethics declaration.} Not applicable.
\bigskip

\noindent\textbf{Consent to participate.} Not applicable.
\bigskip

\noindent\textbf{Consent to publish.} Not applicable.
\bigskip

\noindent\textbf{Data availability.} No datasets were generated or analyzed during this study. Questions regarding the non–database-related aspects of this work may be directed to the corresponding author at yzt0060@auburn.edu.

\printbibliography
\end{document}